\documentclass[a4paper,UKenglish]{my}
 
\usepackage{microtype}

\usepackage{mathtools}
\usepackage[all]{xy}

\theoremstyle{plain}
\newtheorem{proposition}[theorem]{Proposition}

\newcommand*{\Quo}[2]
{\ensuremath{#1\!\raisebox{-.3ex}{/}\!\raisebox{-.65ex}{\ensuremath{#2}}}}

\newcommand*{\EnsembleQuotient}[2]%
{\ensuremath{%
    #1/\!\raisebox{-.65ex}{\ensuremath{#2}}}}
    
\DeclareMathAlphabet{\mathpzc}{OT1}{pzc}{m}{it}
\mathchardef\mhyphen="2D

\title{The Tutte-Grothendieck group of a convergent alphabetic rewriting system}

\author[1]{Laurent Poinsot}
\affil[1]{LIPN - UMR CNRS 7030\\
Universit\'e Paris Nord XIII\\
99 avenue J.-B. Cl\'ement\\
93430 Villetaneuse\\
France\\
  \texttt{laurent.poinsot@lipn.univ-paris13.fr}}
  \authorrunning{L. Poinsot} 


\begin{document}

\maketitle

\begin{abstract}
The two operations, deletion and contraction of an edge, on multigraphs directly lead to the Tutte polynomial which satisfies a universal problem. As observed by  Brylawski~\cite{Brylawski} in terms of order relations,  these operations may be interpreted as a particular instance of a general theory which involves universal invariants like the Tutte polynomial, and a universal group, called the Tutte-Grothendieck group.  In this contribution, Brylawski's theory is extended in two ways: first of all, the order relation is replaced by a string rewriting system, and secondly, commutativity by partial commutations (that permits a kind of interpolation between non commutativity and full commutativity). This allows us to clarify the relations between the semigroup subject to rewriting and the Tutte-Grothendieck group: the later is actually the Grothendieck group completion of the former, up to the free adjunction of a unit (this was even not mention by Brylawski), and normal forms may be seen as universal invariants. Moreover we prove that such universal constructions are also possible in case of a non convergent rewriting system, outside the scope of Brylawski's work.  
 \end{abstract}
 
 \section{Introduction}
 
 In his paper~\cite{Tutte}, Tutte took advantage of two natural operations on (finite multi)graphs (actually on isomorphism classes of multigraphs), deletion and contraction of an edge, in order to introduce the ring $\mathbb{Z}[x,y]$ and a polynomial in two commuting variables $x,y$, also known by Whitney~\cite{Whitney}, unique up to isomorphism since solutions of a universal problem. This polynomial, since called the Tutte polynomial, is a graph invariant in at least two different meanings: first of all, it is defined on isomorphism classes, rather than on actual graphs, in such a way that two graphs with distinct Tutte polynomials are not isomorphic (a well-known functorial point of view), and, secondly, it is invariant with respect to a graph decomposition. Indeed, let $G$ be a graph, and let $e$ be an edge of $G$, which is not a loop (an edge with the same vertex as source and target) nor a bridge (an edge that connects two connected components of a graph).  The edge contraction $G/e$ of $G$ is the graph obtained by identifying the vertices source and target of $e$,  and removing the edge $e$. We write $G-e$ for the graph where the edge $e$ is merely removed; this operation is the edge deletion. Let us consider the graph $G/e+(G-e)$ (well-defined as isomorphic classes) which can be interpreted as a decomposition of $G$. Then, the Tutte polynomial $t$ is invariant with respect to this decomposition in the sense that $t(G)=t(G/e)+t(G-e)$. Moreover this decomposition eventually terminates with graphs with bridges and loops only as edges, and the choice of edges to decompose is irrelevant. 
 
 In his paper~\cite{Brylawski}, Brylawski observed that the previous construction (and many others, for instance the Tutte polynomial for matroids) may be explained in terms of an elegant and  unified categorical framework (namely a universal problem of invariants). In brief, Brylawski considered an abstract notion of decomposition. Let $X$ be a set, and let $\leq$ be an order relation on (a part of) the free commutative semigroup $X^{\oplus}$ (actually Brylawski considered multisets, nevertheless the choice is here made to deal with semigroups since they play a central r\^ole in this contribution), which satisfies a certain number of axioms that are quickly reviewed in informal terms below for the sake of completeness (Appendix~\ref{shortreviewBrylawski} contains a short review of Brylawski's theory in mathematical terms but it may be skipped) and to show how natural is  their translations in terms of rewriting systems.  
 
Let $D(X)$ be a set of formal (finite) sums $\sum_{1\leq i\leq k}n_i x_i$ where $x_i\in X$, $n_i\in\mathbb{N}$ not all of them being zero (an element of the free commutative semigroup $X^{\oplus}$ on $X$) partially ordered by $\leq$.  If $f,g\in D(X)$ such that $f\leq g$, then we say that $f$ \emph{decomposes} into $g$ or that $g$ is a \emph{decomposition} of $f$. Therefore $D(X)$ is seen as a set of commutative \emph{decompositions}. Elements of $X$ that belong to $D(X)$ are assumed to be minimal with respect to $\leq$. Elements of $X\cap D(X)$ that are maximal (and therefore incomparable since also minimal) are said to be \emph{irreducible}. According to a second axiom satisfied by the order relation $\leq$, an element $f\in D(X)$ cannot be decomposed further into any other element of $D(X)$ if, and only if, $f$ is a finite linear combination, with non negative integers as coefficients, of incomparable elements, that is, if $\mathsf{Irr}(X)$ is the set of all irreducibles, then $f$ is not decomposed into another element if, and only if, $f$ is a formal (finite) sum of elements of $\mathsf{Irr}(X)$ with non negative integers as coefficients. This property is similar to the notion of termination in rewriting systems.  Two other properties (\emph{refinability} and \emph{finiteness}) on $D(X)$ ensure that every element of $X$ has one, and only one, "terminal" decomposition into irreducible elements. They are equivalent to convergence of a rewriting system. For instance, the order $G<(G/e)+(G-e)$ on the free commutative semigroup generated by all (isomorphism classes of) finite graphs satisfies these axioms and properties. 
 
 Now, to a decomposition $(D(X),\leq)$ with the above properties may be attached a group in a universal way. A function $f$ from $X$ to an Abelian group $G$ is said to be \emph{invariant} if for every $x\in X$ such that $x\leq \sum_{1\leq i\leq k}n_i x_i$ ($x_i\in X$, and $n_i\in\mathbb{N}$), then $f(x)=\displaystyle n_i f(x_i)$. Recall here that Tutte polynomial $t$ is invariant because $t(G)=t(G/e)+t(G-e)$.  Brylawski proved the following theorem, which was his main result. There exist an Abelian group, called \emph{Tutte-Grothendieck group}, and an invariant mapping $t\colon X\rightarrow A$, called \emph{universal Tutte-Grothendieck invariant}, such that for every Abelian group $G$ and every invariant mapping $f\colon X \rightarrow G$, there exists a unique group homorphism $h\colon A \rightarrow G$ with $h\circ t = f$. In addition, $A$ is isomorphic to the free Abelian group with the irreducible elements as generators. In the classical context of graph theory, as expected, $A$ is the additive structure of $\mathbb{Z}[x,y]$ and $t$ is the Tutte polynomial. Many other decompositions enter in the scope of Brylawski's theory (see his paper~\cite{Brylawski}, examples and references therein).
 
 In the present contribution, we adapt Brylawski's results to the theory of (string) rewriting systems which we think is the natural framework to deal with theoretical notions of decomposition. Moreover we extend previous works by allowing non commutative, and even partially commutative,  decompositions. Our main result, theorem~\ref{mainstatement},  similar to Brylawski's main theorem, states the existence and uniqueness of a universal group and a universal invariant associated to some kind of string rewriting systems, even if there are not convergent (which is beyond the scope of Brylawski's work). In case of convergence, we prove that the universal group under consideration is the free partially commutative group generated by irreducible letters, which is a generalization of the original result, and that the universal invariant is nothing else than the normal form function that maps an element to its normal form.  We mention the fact that in this case, the universal group is proved to be the Grothendieck completion of a monoid (obtained from the semigroup subject to rewriting by free adjunction of an identity), which was not seen by Brylawski even if he called Tutte-Grothendieck his universal construction. 
 
 \section{Some universal constructions}

The categorical notions used in this contribution, that are not defined here, come from~\cite{MacLane}. This section is devoted to the presentation of  Grothendieck group completion and free partially commutative structures which are used here after. 

\subsection{Basic notions and some notations}\label{basics}

In what follows $\mathpzc{S}$, $\mathpzc{M}$ and $\mathpzc{G}$ denote the well-known categories of (small\footnote{"Small" refers to some given fixed universe, see~\cite{MacLane}.}) semigroups, monoids and groups respectively, with their usual arrows (the so-called \emph{homomorphisms of} semigroups, monoids or groups). 

Each of the categories $\mathpzc{S}$, $\mathpzc{M}$ and $\mathpzc{G}$ has a free object freely generated by a given (small) set. In other terms their forgetful functors to the category of sets have a left adjoint. In what follows we denote by $X^{+}$, $X^*$ and $F(X)$ respectively the free semigroup, monoid, group generated by $X$ (see~\cite{BouAlg}), and we identify $X$ as a subset of each of these algebraic structures. Note also that we denote by $X^{\oplus}$ the free commutative semigroup on $X$.

There are also obvious forgetful functors from $\mathpzc{G}$ to $\mathpzc{M}$, and from $\mathpzc{M}$ to $\mathpzc{S}$ (therefore also from $\mathpzc{G}$ to $\mathpzc{S}$ by composition).  Both of them have a left adjoint (see~\cite{MacLane}). The left adjoint of the forgetful functor from $\mathpzc{M}$ to $\mathpzc{S}$ is known to be the free adjunction $S^1=S\sqcup \{\,1\,\}$ of a unit to a semigroup $S$ in order to obtain a monoid in a natural way (the symbol "$\sqcup$" denotes the set-theoretical disjoint sum). The unit of this adjunction, $i_{\mathpzc{S},S}\colon x\in S\rightarrow x\in S^1$, which is an homomorphism of semigroups,  is obviously into. 

The forgetful functor from $\mathpzc{G}\rightarrow \mathpzc{M}$ has both a left and a right adjoint. Its right adjoint is given, at the object level, as a class mapping that associates a monoid to its group of invertible elements. Its left adjoint, more involved, is described below as group completion.

\subsection{Group completion}

The left adjoint of the forgetful functor from groups to monoids may be described as the (unique) solution of the following universal problem. Let $M$ be a monoid. Then there exists a unique group $\mathscr{G}(M)$, called the \emph{group completion} or \emph{universal enveloping group} or \emph{Grothendieck group} of $M$ (see~\cite{Wegge} and references therein, and also~\cite{MayAnno}), and a unique homomorphism of monoids $i_{\mathpzc{M},M}\colon M\rightarrow \mathscr{G}(M)$ such that for every group $G$ and every homomorphism of monoids $f\colon M\rightarrow G$, there is a unique homomorphism of groups $\widehat{f}\colon \mathscr{G}(M)\rightarrow M$ such that the following diagram commutes (in the category of monoids). 
\begin{equation}
\xymatrixcolsep{1pc}
\xymatrix{
M\ar[r]^{f} \ar[d]_{i_{\mathpzc{M},M}} & G\\
\mathscr{G}(M) \ar[ur]_{\widehat{f}} &
}
\end{equation}
It is not difficult to check that $\mathscr{G}(M)$ is given either as $\Quo{F(M)}{\langle I_M\rangle}$ where $I_M$ is the subset $\{\, mn(m*n)^{-1}\colon m,n\in M\,\}$ (where "$*$" is the monoid multiplication of  $M$, and where $F(X)$ denotes the free group of $X$, see Subsection~\ref{basics} and if $G$ is a group and $A$ is any subset of $G$, then $\langle A \rangle$ is the normal subgroup of $G$ generated by $A$), see~\cite{MayAnno}, or as the quotient monoid $\Quo{(M\sqcup M^{-1})^*}{\equiv_R}$ where $R=\{\, (mm^{-1},\epsilon)\colon m \in M\,\} \cup\{\,(m^{-1}m,\epsilon)\colon m\in M \,\}$ (here the star "$*$" stands for the free monoid functor, see also Subsection~\ref{basics}, and $\epsilon$ is the empty word) where $M^{-1}$ is the set of (formal) symbols $\{\, m^{-1}\colon m\in M\,\}$ equipotent to $M$.

\subsection{Free partially commutative structures}\label{fpcs}

Other universal problems, which will play an important r\^ole in what follows, are the free partially commutative structures. These structures have been introduced in~\cite{CartierFoata} (see also~\cite{Viennot}). A good review of these objects is~\cite{GHEDKrob}. Since such constructions may be performed in any of the categories of semigroups, monoids and groups, they are presented here in a generic way on a category $\mathpzc{C}\in\{\, \mathpzc{S},\mathpzc{M},\mathpzc{G}\,\}$ so that all statements make sense in any of these categories.

Let $X$ be a set and let $\theta\subseteq X\times X$ be a symmetric ({\it i.e.}, for every $x,y\in X$, $(x,y)\in\theta$ implies $(y,x)\in \theta$) and reflexive relation on $X$ ({\it i.e.}, for each $x\in X$, $(x,x)\not\in \theta$). Let $C$ be an object in $\mathpzc{C}$, and $f\colon X\rightarrow C$ be a set-theoretical mapping. This function is said to \emph{respect the commutations} whenever $(x,y)\in \theta$ then $f(x)f(y)=f(y)f(x)$, for every $x,y\in X$. A pair $(X,\theta)$ is called a {\it commutation alphabet}. 

It can be shown that there exists a unique object $\mathpzc{C}(X,\theta)$ of $\mathpzc{C}$ and a unique mapping $j_{\mathpzc{C},X}\colon X \rightarrow C(X,\theta)$ that respects the commutations such that for every object $C$ of $\mathpzc{C}$ and every mapping $f\colon X\rightarrow C$ that respects the commutations, there is a unique arrow (in $\mathpzc{C}$) $f^{\mathpzc{C}}\colon \mathpzc{C}(X,\theta)\rightarrow C$ such that the following diagram commutes in the category of sets. 
\begin{equation}
\xymatrixcolsep{1pc}
\xymatrix{
X\ar[r]^{f} \ar[d]_{j_{\mathpzc{C},X}} & C\\
\mathpzc{C}(X,\theta) \ar[ur]_{f^{\mathpzc{C}}} &
}
\end{equation}
The object $\mathpzc{C}(X,\theta)$ is usually called the \emph{free partially commutative semigroup} (respectively, \emph{monoid}, \emph{group}) on $X$ (or on $(X,\theta)$ to be more precise) depending on $\mathpzc{C}$, and may be constructed as follows: $\mathpzc{S}(X,\theta)=\Quo{X^{+}}{\equiv_{\theta}}$ and $\mathpzc{M}(X,\theta)=\Quo{X^*}{\equiv_{\theta}}$ where $\equiv_{\theta}$ is the \emph{congruence} on $X^{+}$ or $X^{*}$ \emph{generated by} $(xy,yx)$ whenever $(x,y)\in \theta$ for all $x,y\in X$ (the least congruence on $X^{+}$ or $X^*$ containing the relation $(xy,yx)$ whenever $(x,y)\in \theta$ for all $x,y\in X$, see~\cite{Cliff}), and $\mathpzc{G}(X,\theta)=\Quo{F(X,\theta)}{\langle \{\,xyx^{-1}y^{-1}\colon (x,y) \in \theta\,\} \rangle}$. 

We may note that $\mathpzc{C}(X,\emptyset)$ is nothing else than the usual free (non commutative) object in the category $\mathpzc{C}$, while $\mathpzc{C}(X,(X\times X)\setminus \Delta_X)$, where $\Delta_X$ is the equality relation on $X$, is the free commutative object in $\mathpzc{C}$ (in particular, $\mathpzc{S}(X,(X\times X)\setminus \Delta_X)=X^{\oplus}$ is the free commutative semigroup).


We may clarify the relations between the free partially commutative structures. Using universal properties, it is not difficult to check that  $\mathpzc{M}(X,\theta)$ is isomorphic to $\mathpzc{S}(X,\theta)^1$  (actually $\mathpzc{M}(X,\theta)=\mathpzc{S}(X,\theta)\cup\{\,\epsilon\,\}$, where $\epsilon$ is the empty word) in such a way that $\mathpzc{S}(X,\theta)$ embeds in $\mathpzc{M}(X,\theta)$ as a sub-semigroup. 
\begin{lemma}\label{adjunctionunitfpcs}
The monoid $\mathpzc{M}(X,\theta)$ is isomorphic to the free adjunction $\mathpzc{S}(X,\theta)^1$ of an identity to the semigroup $\mathpzc{S}(X,\theta)$.
\end{lemma}
\begin{proof}
To prove this lemma it is sufficient to check that $\mathpzc{M}(X,\theta)$ is a solution of the universal problem of adjunction of a unit to $\mathpzc{S}(X,\theta)$. According to the universal problem of the free partially commutative semigroup $\mathpzc{S}(X,\theta)$, there is a unique homomorphism of semigroups $I\colon \mathpzc{S}(X,\theta)\rightarrow \mathpzc{M}(X,\theta)$ such that the following diagram is commutative.
\begin{equation}
\xymatrixcolsep{1pc}
\xymatrix{
X\ar[r]^{id_X} \ar[d]_{j_{\mathpzc{S},X}}& X \ar[d]^{j_{\mathpzc{M},X}}\\
\mathpzc{S}(X,\theta) \ar[r]_{I} & \mathpzc{M}(X,\theta)
}
\end{equation}
Now, let $M$ be a monoid and $f\colon \mathpzc{S}(X,\theta)\rightarrow M$ be a semigroup homomorphism. Therefore there exists $f_0\colon X\rightarrow M$ that respects the commutations and such that $f_0^{\mathpzc{S}}\circ j_{\mathpzc{S},X}=f$. According to the universal problem attached to $\mathpzc{M}(X,\theta)$, there is a unique homomorphism of monoids $f_0^{\mathpzc{M}}\colon \mathpzc{M}(X,\theta)\rightarrow f_0$ such that $f_0^{\mathpzc{M}}\circ j_{\mathpzc{M},X}=f_0$. Therefore, $f_0^{\mathpzc{M}}\circ I\circ  j_{\mathpzc{S},X}=f^0$, but then $f_0^{\mathpzc{M}}\circ I=f_0^{\mathpzc{S}}=f$. The relations between all the arrows are summarized in the following commutative diagram. 
\begin{equation}
\xymatrixcolsep{1pc}
\xymatrix{
X\ar[rr]^{id_X} \ar[d]^{j_{\mathpzc{S},X}}\ar@/_3.3pc/[ddr]_{f_0}& &X \ar[d]^{j_{\mathpzc{M},X}}\\
\mathpzc{S}(X,\theta) \ar[rr]_{I} \ar[rd]_{f}& &\mathpzc{M}(X,\theta)\ar[ld]^{f_0^{\mathpzc{M}}}\\
& M &
}
\end{equation}
\end{proof}
There is also an important relation between $\mathpzc{G}(X,\theta)$ and $\mathpzc{M}(X,\theta)$ given in the following lemma. 
\begin{lemma}\label{Grothendieckoffpcg}
Let $(X,\theta)$ be a commutation alphabet. Then, $\mathpzc{G}(X,\theta)$ is (isomorphic to) the universal enveloping group $\mathscr{G}(\mathpzc{M}(X,\theta))$ of $\mathpzc{M}(X,\theta)$. 
\end{lemma}
\begin{proof}
The set-theoretical mapping $j_{\mathpzc{G},X}\colon X \rightarrow \mathpzc{G}(X,\theta)$ respects the commutations, therefore according to the universal problem of the free partially commutative monoid over $(X,\theta)$ there is a unique homomorphism of monoids $j_{\mathpzc{G},X}^{\mathpzc{M}}$ that makes commute the following diagram. 
\begin{equation}
\xymatrixcolsep{1pc}
\xymatrix{
X\ar[r]^{j_{\mathpzc{G},X}} \ar[d]_{j_{\mathpzc{M},X}} & \mathpzc{G}(X,\theta)\\
\mathpzc{M}(X,\theta) \ar[ur]_{j_{\mathpzc{G},X}^{\mathpzc{M}}} &
}
\end{equation}
Now, let $G$ be any group, and $f\colon \mathpzc{M}(X,\theta)\rightarrow G$ be an homomorphism of monoids. Then, according to the universal problem of the free partially commutative monoid, there is a unique set-theoretical mapping $f_0\colon X \rightarrow G$ that respects the commutations and $f\circ j_{\mathpzc{M},X}=f_0$. Now according to universal problem of $\mathpzc{G}(X,\theta)$, $f_0$ is uniquely extended as a group homomorphism $f_0^{\mathpzc{G}}\colon \mathpzc{G}(X,\theta)\rightarrow G$ such that $f_0^{\mathpzc{G}}\circ j_{\mathpzc{G},X}=f_0$. Therefore, $f_0^{\mathpzc{G}}\circ j_{\mathpzc{G},X}^{\mathpzc{M}}\circ j_{\mathpzc{M},X}=f_0^{\mathpzc{G}}\circ j_{\mathpzc{G},X}=f_0=f\circ j_{\mathpzc{M},X}$ so that $f_0^{\mathpzc{G}}\circ j_{\mathpzc{G},X}^{\mathpzc{M}}=f$ (by uniqueness of a solution of a universal problem). Therefore $(\mathpzc{G}(X,\theta),f_0^{\mathpzc{G}})$ is a solution of the universal problem of the group completion $\mathscr{G}(\mathpzc{M}(X,\theta))$ of $\mathpzc{M}(X,\theta)$. The relations between all the arrows are summarized in the following commutative diagram. 
\begin{equation}
\xymatrixcolsep{1pc}
\xymatrix{
X\ar[r]^{j_{\mathpzc{G},X}} \ar[d]_{j_{\mathpzc{M},X}}\ar@/_2.5pc/[dd]_{f_0} & \mathpzc{G}(X,\theta) \ar@/^1pc/[ldd]^{f^{\mathpzc{G}}_0}\\
\mathpzc{M}(X,\theta) \ar[d]_{f} \ar[ur]_{j_{\mathpzc{G},X}^{\mathpzc{M}}} &\\
G & 
}
\end{equation}
\end{proof}
Actually a result from~\cite{Db} page 66 (see also~\cite{GHEDKrob}) states that the natural mapping $j_{\mathpzc{G},X}^{\mathpzc{M}}$ of the proof of lemma~\ref{Grothendieckoffpcg} is one-to-one so that $\mathpzc{M}(X,\theta)$ may be identified with a sub-monoid of its Grothendieck completion $\mathpzc{G}(X,\theta)$. 

\begin{definition}
Let $X$ be any set. For every $x\in X$ and every $w\in X^{*}$, let us define $|w|_x$ as the number of occurrences of the letter $x$ in the word $w$. More precisely, if $\epsilon$ is the empty word, then $|\epsilon|_x=0$, $|y|_x=0$ if $y\not=x$, $|y|_x=1$ if $y=x$ for all $y\in X$, and if the length of $w\in X^{*}$ is $>1$, then $w=yw^{\prime}$ for some letter $y\in X$, and $w^{\prime}\in X^{+}$, then $|w|_x=|y|_x+|w^{\prime}|_x$. Let $\equiv$ be a congruence on $X^{+}$ or $X^{*}$. It is said to be \emph{multi-homogeneous} if for every $w,w^{\prime}$ in $X^{+}$ or $X^{*}$, such that $w\equiv w^{\prime}$, then for every $x\in X$, $|w|_x=|w^{\prime}|_x$. Therefore we may define $|[w]_{\equiv}|_x=|w|_x$ for the class $[w]_{\equiv}$ of $w$ modulo $\equiv$ (it does not depend on the representative of the class modulo $\equiv$).
\end{definition}
According to~\cite{GHEDKrob}, any congruence of the form $\equiv_{\theta}$ is a multi-homogenous congruence, so that we may define $|w|_x$ for all $w\in \mathpzc{C}(X,\theta)$ and all $x\in X$ (where $\mathpzc{C}=\mathpzc{S}$ or $\mathpzc{M}$). The notion of multi-homogeneity is used to check that we may identify the alphabet $X$ has a generating set of $\mathpzc{C}(X,\theta)$ using the map $j_{\mathpzc{C},X}$, which is shown to be into, in such a way that we consider that $X\subseteq \mathpzc{C}(X,\theta)$.  Indeed, for semigroup or monoid case, let $x,y\in X$ such that their classes modulo $\equiv_{\theta}$ be equal. But $\equiv_{\theta}$ is a multi-homogenous congruence (see~\cite{GHEDKrob}). Therefore $x=y$. Concerning the group case, let us assume that $x,y\in X$ are equivalent modulo the normal subgroup $N_{\theta}=\langle \{\,xyx^{-1}y^{-1}\colon (x,y) \in \theta\,\} \rangle$ so that there is some $w\in N_{\theta}$ with $xy^{-1}=w$. Because the group is free, it means that $x=y$ (no non trivial relations between the generators). In the sequel, we will treat $X$ as a subset of $\mathpzc{C}(X,\theta)$. 

More generally, let $(X,\theta)$ be a commutation alphabet and let $Y\subseteq X$. We define $\theta_Y=\theta\cap (Y\times Y)$. It is possible to embed $\mathpzc{C}(Y,\theta_Y)$ into $\mathpzc{C}(X,\theta)$ as illustrated in the following lemma.
\begin{lemma}\label{identification}
Under the previous assumptions, there is an arrow $J\colon \mathpzc{C}(Y,\theta_Y)\rightarrow \mathpzc{C}(X,\theta)$ in the category $\mathpzc{C}$ which is into. 
\end{lemma}
\begin{proof}
Let $\mathsf{incl}\colon Y \rightarrow X$ be the canonical inclusion. Define $J\colon \mathpzc{C}(Y,\theta_Y)\rightarrow \mathpzc{C}(X,\theta)$ as the unique arrow (in $\mathpzc{C}$) such that the following diagram commutes.
\begin{equation}
\xymatrixcolsep{2pc}
\xymatrix{
Y\ar[r]^{\mathsf{incl}} \ar[d]_{j_{\mathpzc{C},Y}} & X \ar[d]_{j_{\mathpzc{C},X}}\\
\mathpzc{C}(Y,\theta_Y) \ar[r]_{J}& \mathpzc{C}(X,\theta)}
\end{equation}
Therefore, $J\circ j_{\mathpzc{C},Y}=j_{\mathpzc{C},X}\circ\mathsf{incl}$. 

Let $w_0\in \mathpzc{C}(Y,\theta_Y)$. Let us define $\pi_{w_0}\colon X \rightarrow Y$ such that $\pi_{w_0}(y)=y$ for every $y\in Y\subseteq X$, and $\pi_{w_0}(x)=w_0$ for $x\in X\setminus Y$. We note that $\pi_{w_0}\circ \mathsf{incl}=id_Y$. Then we may consider $\Pi_{w_0}\colon \mathpzc{C}(X,\theta)\rightarrow \mathpzc{C}(Y,\theta_Y)$ as the unique arrow (in $\mathpzc{C}$) that makes commute the following diagram. 
\begin{equation}
\xymatrixcolsep{2pc}
\xymatrix{
X\ar[r]^{\pi_{w_0}} \ar[d]_{j_{\mathpzc{C},X}} & Y \ar[d]_{j_{\mathpzc{C},Y}}\\
\mathpzc{C}(X,\theta) \ar[r]_{\Pi_{w_0}}& \mathpzc{C}(Y,\theta_Y)}
\end{equation}
Therefore $\Pi_{w_0}\circ j_{\mathpzc{C},X}=j_{\mathpzc{C},Y}\circ \pi_{w_0}$. Now, $\Pi_{w_0}\circ J\circ j_{\mathpzc{C},Y}=\Pi_{w_0}\circ j_{\mathpzc{C},X}\circ \mathsf{incl}=j_{\mathpzc{C},Y}\circ \pi_{w_0}\circ \mathsf{incl}=j_{\mathpzc{C},Y}\circ id_{Y}=j_{\mathpzc{C},Y}=id_{\mathpzc{C}(Y,\theta_Y)}\circ j_{\mathpzc{C},Y}$, so that (by uniqueness) $\Pi_{w_0}\circ J=id_{\mathpzc{C}(Y,\theta_Y)}$, and then $J$ is into (and $\Pi_{w_0}$ is onto). 
\end{proof}
According to lemma~\ref{identification} we identify $\mathpzc{C}(Y,\theta_Y)$ as a sub-semigroup, sub-monoid or sub-group (depending on the choice of $\mathpzc{C}$) of $\mathpzc{C}(X,\theta)$. In such situations we may use the following characterization.
\begin{lemma}\label{characterization}
Let $(X,\theta)$ be a commutation alphabet, and let $Y\subseteq X$ be any subset. Let $w\in \mathpzc{S}(X,\theta)$. The following statements are equivalent:
\begin{enumerate}
\item $w\in \mathpzc{S}(Y,\theta_Y)$.
\item For all $x\in X$, $|w|_x\not=0$ implies that $x\in Y$. 
\end{enumerate} 
\end{lemma} 
\begin{proof}
Let $w\in \mathpzc{S}(X,\theta)$. If $w\in \mathpzc{S}(Y,\theta_Y)$, then for all $\omega\in X^{+}$ such that $\omega\in w$, $\omega \in Y^{+}$ (since $\mathpzc{S}(Y,\theta_Y)\cong \Quo{Y^{+}}{\equiv_{\theta}}$). Because $\equiv_{\theta}$ is a multi-homogeneous congruence, $|\omega|_x=|w|_x$ for all $\omega\in w$ and $x\in X$. Then the point 2. is obtained. Now, let $w\in \mathpzc{S}(X,\theta)$ such that for all $x\in X$, $|w|_x\not=0$ implies that $x\in Y$. Then, for all $\omega\in w$ ($\omega\in X^{+}$), $|\omega|_x=0$ for all $x\not\in Y$ which means that $\omega\in Y^{+}$, and therefore $w\in \mathpzc{S}(Y,\theta_Y)$ so that 1. is obtained.
\end{proof}

\section{Basic on rewriting systems}

\subsection{Abstract rewriting systems}\label{ars}
In this short section, as in the following, we adopt several notations and definitions from~\cite{BaaNip} that we summarize here. 
 
Let $E$ be a set, and $\Rightarrow\subseteq  E\times E$ be any binary relation, called a \emph{(one-step) reduction relation}, and $(E,\Rightarrow)$ is called an \emph{abstract rewriting system}. We denote by "$x\Rightarrow y$" the membership "$(x,y)\in\Rightarrow$", and "$x\not\Rightarrow y$" stands for "$(x,y)\not\in \Rightarrow$". Let $R^*$ be the reflexive transitive closure of a binary relation $R$.   We use $x\Leftarrow y$ or $x\xLeftarrow{*} y$ to mean that $y\Rightarrow x$ or $y\xRightarrow{*} x$. An element $x\in E$ is said to be \emph{reducible} if there exists $y\in E$ such that $x\Rightarrow y$. $x$ is \emph{irreducible} if it is not reducible, or, in other terms, if $x$ is \emph{$\Rightarrow$-minimal}: there is no $y\in E$ such that $x\Rightarrow y$.  A \emph{normal form} of $x$ is an irreducible element $y\in E$ such that $x\xRightarrow{*} y$. If it exists, the normal form of $x$ is denoted by $\mathpzc{N}(x)$. The set of all normal forms, or equivalently, of all irreducible elements is denoted by $\mathsf{Irr}(E,\Rightarrow)$ or $\mathsf{Irr}(E)$ when this causes no ambiguity. Note that two distinct normal forms $x,y$ are \emph{$\Rightarrow$-incomparable}, that is $x\not\Rightarrow y$ and $y\not\Rightarrow x$. A reduction relation $\Rightarrow$ is said to be \emph{terminating} or \emph{Noetherian} if there is no infinite $\Rightarrow$-descending chain $(x_n)_{n\in\mathbb{N}}$ of elements of $E$ such that $x_n\Rightarrow x_{n+1}$ for every $n\geq 0$. In particular, if $\Rightarrow$ is terminating, then it is irreflexive (otherwise $x_n=x$ for some $x\in E$ such that $x\Rightarrow x\in R$ would be an infinite $\Rightarrow$-descending chain), that is the reason why we freely make use terminology from order relations (such as minimal, Noetherian, descending chain, {\it{etc.}}).  We also say that the abstract rewriting system $(E,\rightarrow)$ is \emph{terminating} or \emph{Noetherian} whenever $\Rightarrow$ is so. Two elements $x,y\in E$ are said to be \emph{joignable} if there is some $z\in E$ such that  $x\xRightarrow{*} z \xLeftarrow{*} y$, and  $\Rightarrow$ (and also $(E,\Rightarrow$)) is said to be \emph{confluent} if for every $x,y_1,y_2\in E$ such that $y_1 \xLeftarrow{*} x \xRightarrow{*}y_2$, then $y_1,y_2$ are joignable. A reduction relation $\Rightarrow$, and an abstract rewriting system $(E,\Rightarrow)$, are said to be \emph{convergent} if it they are both confluent and terminating. Such reduction relations are interesting because in this case any element of $E$ has one, and only one, normal form, and if we denote by $\xLeftrightarrow{*}$ the reflexive transitive symmetric closure of $\Rightarrow$ (that is the least equivalence relation on $E$ containing $\Rightarrow$), then $x\xLeftrightarrow{*}y$ if, and only if, $\mathpzc{N}(x)=\mathpzc{N}(y)$, therefore $\mathpzc{N}\colon E \rightarrow \mathsf{Irr}(E)$ satisfies $\mathpzc{N}(\mathpzc{N}(x))=\mathpzc{N}(x)$ and so is onto and moreover, the function $\overline{\mathpzc{N}}\colon \Quo{E}{\xLeftrightarrow{*}}\rightarrow \mathsf{Irr}(E)$ which maps the class of $x$ modulo  $\xLeftrightarrow{*}$ to $\mathpzc{N}(x)$ is well-defined, onto and one-to-one.  

\subsection{Semigroup rewriting systems}

Now, let us assume that $E$ is actually a semigroup $S$. Let $R\subseteq S\times S$ be any binary relation. We define the following relation $\Rightarrow_R\subseteq S\times S$ by $x\Rightarrow_R y$ if, and only if, there are $u,v\in S^1$ and $(a,b)\in R$ such that $x=uav$ and $y=ubv$. A relation $\Rightarrow_R$ is called the \emph{(one-step) reduction rule} associated with $R$. A relation $R\subseteq S\times S$ is said to be \emph{two-sided compatible} if $(x,y)\in R$ ($x,y\in S$) implies $(uxv,uyv)\in R$. Now, the intersection of the family of all two-sided compatible relations containing a given $R\subseteq S\times S$ (this family is non void since it contains the universal relation $S\times S$) also is a two-sided compatible relation, and so we obtain the least two-sided compatible relation that contains $R$. It is called the \emph{two-sided compatible relation generated by $R$}, and it can be shown that this is precisely $\Rightarrow_R$. Now, given $R\subseteq S\times S$,   $(S,\Rightarrow_R)$ is called a \emph{(semigroup) rewriting system}; definitions and properties of an abstract rewriting system may be applied to such a rewriting system. When $S$ is the free monoid $X^*$, then this kind of rewriting systems are known as \emph{string rewriting systems} or \emph{semi-Thue systems} (see~\cite{BookOtto}). We note that the reflexive transitive symmetric closure $\xLeftrightarrow{*}_R$ of $\Rightarrow_R$ is actually a semigroup congruence, because $\Rightarrow_R$ is two-sided compatible. The quotient semigroup $\Quo{S}{\xLeftrightarrow{*}_R}$ is called the \emph{Thue semigroup} associated with the semigroup rewriting system $(S,\Rightarrow_R)$.

\section{The Tutte-Grothendieck group of a convergent alphabetic rewriting system}

\subsection{A free partially commutative structure on normal forms}\label{normalformsarefpcs}
\begin{definition}\label{alphabetic}
Let $(X,\theta)$ be a commutation alphabet, and $R\subseteq X\times \mathpzc{S}(X,\theta)$. Then $(\mathpzc{S}(X,\theta),\Rightarrow_R)$ is called an \emph{alphabetic} semigroup rewriting system. 
\end{definition}

From now on in this current subsection we assume that $(\mathpzc{S}(X,\theta),\Rightarrow_R)$ is a convergent alphabetic semigroup rewriting system. 

We study some algebraic consequences of convergence of this alphabetic rewriting system on irreducible elements in the form of some lemmas and corollaries. The main result (proposition~\ref{Freepartcomstruonirr}) of this subsection is that the set of all normal forms of a convergent alphabetic semigroup rewriting system is actually the free partially commutative semigroup in a canonical way, generated by the irreducible letters. 
\begin{lemma}\label{productofnormalforms}
Let $w,w^{\prime}\in \mathsf{Irr}(\mathpzc{S}(X,\theta),\Rightarrow_R)$. Then, $ww^{\prime}\in\mathsf{Irr}(\mathpzc{S}(X,\theta),\Rightarrow_R)$. As a result, $\mathsf{Irr}(\mathpzc{S}(X,\theta),\Rightarrow_R)$ is a sub-semigroup of $\mathpzc{S}(X,\theta)$. 
\end{lemma}
\begin{proof}
Let us assume that $ww^{\prime}\not\in \mathsf{Irr}(\mathpzc{S}(X,\theta),\Rightarrow_R)$. Therefore there are $x\in X$, $w^{\prime\prime}, w^{\prime\prime\prime}\in \mathpzc{S}(X,\theta)$, $u,v\in\mathpzc{M}(X,\theta)$ such that $(x,w^{\prime\prime\prime})\in R$, $ww^{\prime}=uxv$ and $w^{\prime\prime}=uw^{\prime\prime\prime}v$ (so that $ww^{\prime}\Rightarrow_R w^{\prime\prime}$). Because $\equiv_{\theta}$ is multi-homogeneous, either $w=u^{\prime}xv^{\prime}$ or $w^{\prime}=u^{\prime}xv^{\prime}$ for some $u^{\prime},v^{\prime}\in \mathpzc{M}(X,\theta)$. But in this case, either $w$ or $w^{\prime}$ is reducible, which is a contradiction. As a result, $\mathsf{Irr}(\mathpzc{S}(X,\theta),\Rightarrow_R)\subseteq \mathpzc{S}(X,\theta)$ is closed under the operation of $\mathpzc{S}(X,\theta)$ so that $\mathsf{Irr}(\mathpzc{S}(X,\theta),\Rightarrow_R)$ is  a sub-semigroup of $\mathpzc{S}(X,\theta)$.
\end{proof}
\begin{corollary}\label{normalformishomomorphism}
The map $\mathpzc{N}\colon \mathpzc{S}(X,\theta)\rightarrow \mathsf{Irr}(\mathpzc{S}(X,\theta),\Rightarrow_R)$ is a surjective homomorphism of semigroups. 
\end{corollary}
\begin{proof}
Let $w,w^{\prime}\in \mathpzc{S}(X,\theta)$. According to lemma~\ref{productofnormalforms}, $\mathpzc{N}(w)\mathpzc{N}(w^{\prime})\in \mathsf{Irr}(\mathpzc{S}(X,\theta),\Rightarrow_R)$. Therefore, $\mathpzc{N}(\mathpzc{N}(w)\mathpzc{N}(w^{\prime}))=\mathpzc{N}(w)\mathpzc{N}(w^{\prime})$. Since $\xLeftrightarrow{*}$ is a congruence of $\mathpzc{S}(X,\theta)$, $ww^{\prime}\xLeftrightarrow{*}\mathpzc{N}(w)\mathpzc{N}(w^{\prime})$ in such a way that $\mathpzc{N}(ww^{\prime})=\mathpzc{N}(\mathpzc{N}(w)\mathpzc{N}(w^{\prime}))=\mathpzc{N}(w)\mathpzc{N}(w^{\prime})$ and $\mathpzc{N}$ is an homomorphism of semigroups. It is obviously onto.
\end{proof}
\begin{corollary}\label{isomorphism}
The semigroups $\mathsf{Irr}(\mathpzc{S}(X,\theta),\Rightarrow_R)$ and $\Quo{\mathpzc{S}(X,\theta)}{\xLeftrightarrow{*}}$ are isomorphic.
\end{corollary}
\begin{proof}
As introduced in Subsection~\ref{ars}, let $\overline{\mathpzc{N}}\colon \Quo{\mathpzc{S}(X,\theta)}{\xLeftrightarrow{*}} \rightarrow \mathsf{Irr}(\mathpzc{S}(X,\theta),\Rightarrow_R)$ be the function that maps the class of $w\in \mathpzc{S}(X,\theta)$ modulo $\xLeftrightarrow{*}$ to the normal form $\mathpzc{N}(w)$. It is a one-to-one and onto set-theoretical mapping. But according to corollary~\ref{normalformishomomorphism}, $\mathpzc{N}$ is a semigroup homomorphism, in such a way that $\overline{\mathpzc{N}}$ also is.
\end{proof}

The fact that the rewriting system is alphabetic (Definition~\ref{alphabetic}) actually implies that the (isomorphic) semigroups $\mathsf{Irr}(\mathpzc{S}(X,\theta),\Rightarrow_R)$ and $\Quo{\mathpzc{S}(X,\theta)}{\xLeftrightarrow{*}}$ are actually free partially commutative. The objective is now to prove this statement. In order to do that, we exhibit the commutation alphabet that generates them. 
Let $\mathsf{Irr}(X)=\mathsf{Irr}(\mathpzc{S}(X,\theta),\Rightarrow_R)\cap X$ (recall from Subsection~\ref{fpcs} that $X$ is considered as a subset of $\mathpzc{S}(X,\theta)$). It is clear that $\mathsf{Irr}(X)=\{\, x\in X\colon \nexists w\in S(X,\theta),\ (x,w)\in R\,\}$. Indeed, for every $x\in X$, $w\in \mathpzc{S}(X,\theta)$, $x\Rightarrow_R w$ if, and only if, there are $u,v\in\mathpzc{M}(X,\theta)$, $x_1\in X$, $w_1\in \mathpzc{S}(X,\theta)$ such that $x=ux_1v$ and $w=uw_1v$. Since $\equiv_{\theta}$ is a multi-homogenous congruence (see subsection~\ref{fpcs}), $u=v$ is the empty word, and $x=x_1$, $w=w_1$. Therefore $x\Rightarrow_R w$ if, and only if, $(x,w)\in R$. 

This characterization of $\mathsf{Irr}(X)$ is used in the following lemma.
\begin{lemma}
If $X\not=\emptyset$, then $\mathsf{Irr}(X)\not=\emptyset$.
\end{lemma}
\begin{proof}
Let us assume that $X\not=\emptyset$ and $\mathsf{Irr}(X)=\emptyset$. Let $x\in X$. Since $x\not\in\mathsf{Irr}(X)$, there is some $w\in \mathpzc{S}(X,\theta)$ such that $(x,w)\in R$. Because $w\in\mathpzc{S}(X,\theta)$, and $X$ generates $\mathpzc{S}(X,\theta)$, it can be written as $x_1u$ for some $x_1\in X$, and $u\in \mathpzc{M}(X,\theta)$. Because $x_1\not\in \mathsf{Irr}(X)$, there is $v_1\in \mathpzc{S}(X,\theta)$ such that $(x_1,v_1)\in R$. Then, $w\Rightarrow_R v_1u$. Replacing $w$ by $v_1u$, we may construct an infinite descending chain $x\Rightarrow_R w \Rightarrow_R v_1u\Rightarrow_R \cdots$, which is impossible since $\Rightarrow_R$ is assumed to be convergent, and therefore terminating. So $\mathsf{Irr}(X)\not=\emptyset$.
\end{proof}
\begin{remark}
Forthcoming proposition~\ref{Freepartcomstruonirr}, lemmas~\ref{comportementdesRinvariants1} and~\ref{comportementdesRinvariants2}, and theorem~\ref{mainstatement} are obviously valid when $X=\emptyset$.
\end{remark}
The following lemma reveals the structure of free partially commutative semigroup of $\mathsf{Irr}(\mathpzc{S}(X,\theta),\Rightarrow_R)$, and therefore also of $\Quo{\mathpzc{S}(X,\theta)}{\xLeftrightarrow{*}}$ according to lemma~\ref{isomorphism}.
\begin{proposition}\label{Freepartcomstruonirr}
The semigroup $\mathsf{Irr}(\mathpzc{S}(X,\theta),\Rightarrow_R)$ is equal to the free partially commutative semigroup $\mathpzc{S}(\mathsf{Irr}(X),\theta_{\mathsf{Irr}(X)})$ where $\theta_{\mathsf{Irr}(X)}=\theta\cap(\mathsf{Irr}(X)\times\mathsf{Irr}(X))$ (see Lemma~\ref{identification}).
\end{proposition}
\begin{proof}
Let $w\in \mathsf{Irr}(\mathpzc{S}(X,\theta),\Rightarrow_R)$. Let us assume that $w\not\in \mathpzc{S}(\mathsf{Irr}(X),\theta_{\mathsf{Irr}(X)})$. According to lemma~\ref{characterization}, there exists $x\in X\setminus\mathsf{Irr}(X)$ such that for all $\omega\in X^{+}$, $\omega\in w$ ($w$ is seen as a congruence class), $|\omega|_x\not=0$. Therefore $\omega = uxv$ for some $u,v\in X^{*}$ and $w=\pi_{\theta}(u)x\pi_{\theta}(v)$ (where $\pi_{\theta}\colon X^{*}\rightarrow \mathpzc{M}(X,\theta)$ is the canonical epimorphism and where we recall that $X$ is seen as a subset of $\mathpzc{S}(X,\theta)$, $X^{*}=X^{+}\sqcup\{\,\epsilon\,\}$, and $\mathpzc{M}(X,\theta)=\mathpzc{S}(X,\theta)\sqcup\{\,\epsilon\,\}$). But $x\not\in\mathsf{Irr}(X)$, then there exists $w^{\prime}\in \mathpzc{S}(X,\theta)$ such that $(x,w^{\prime})\in R$, and therefore $w\Rightarrow_R \pi_{\theta}(u)w^{\prime}\pi_{\theta}(v)$ which contradicts the fact that $w\in \mathsf{Irr}(\mathpzc{S}(X,\theta),\Rightarrow_R)$. Let $w\in \mathpzc{S}(X,\theta)$ such that $w\in\mathpzc{S}(\mathsf{Irr}(X),\theta_{\mathsf{Irr}(X)})$. Let us assume that $w\not\in \mathsf{Irr}(\mathpzc{S}(X,\theta),\Rightarrow_R)$. Therefore $w=uxv$ for some $u,v\in\mathpzc{M}(X,\theta)$, $x\in X$ such that there is $w^{\prime}\in \mathpzc{S}(X,\theta)$ with $(x,w^{\prime})\in R$. Therefore $x\not\in\mathsf{Irr}(X)$. It is then clear that for every $\omega\in X^{+}$ such that $\omega\in w$, $|\omega|_x>0$. But according to lemma~\ref{characterization}, this is impossible because $x\not\in\mathsf{Irr}(X)$ and $w\in \mathpzc{S}(\mathsf{Irr}(X),\theta_{\mathsf{Irr}(X)})$. We have proved that $\mathsf{Irr}(\mathpzc{S}(X,\theta),\Rightarrow_R)$ and $\mathpzc{S}(\mathsf{Irr}(X),\theta_{\mathsf{Irr}(X)})$ are equal as sets. But since they are both sub-semigroups of $\mathpzc{S}(X,\theta)$, then they are equal as semigroups.
\end{proof}

\subsection{The Tutte-Grothendieck group of a convergent alphabetic rewriting system}\label{TGgroup}

\begin{definition}\label{R-invariant}
Let $(X,\theta)$ be a commutation alphabet, and let $\Rightarrow_R$ be an alphabetic rewriting system.  Let $S$ be any semigroup, and let $f\colon X\rightarrow S$ that respects the commutations. Let $f^{\mathpzc{S}}\colon \mathpzc{S}(X,\theta)\rightarrow S$ be the unique homomorphism of semigroups such that the following diagram commutes (see Subsection~\ref{fpcs}). 
\begin{equation}
\xymatrixcolsep{1pc}
\xymatrix{
X\ar[r]^{f} \ar[d]_{j_{\mathpzc{S},X}} & S\\
\mathpzc{S}(X,\theta) \ar[ur]_{f^{\mathpzc{S}}} &
}
\end{equation}
Then $f$ is said to be an \emph{$R$-invariant} if for every $x\in X$ and $w\in \mathpzc{S}(X,\theta)$ such that $(x,w)\in R$, then $f(x)=f^{\mathpzc{S}}(w)$. 
\end{definition}
Informally speaking, according to definition~\ref{R-invariant}, a function $f$ that respects the commutations is an $R$-invariant if its canonical semigroup extension $f^{\mathpzc{S}}$ is constant  for all reductions $(x,w)\in R$.

Let us assume that $(X,\theta)$ is a commutation alphabet, and let $\Rightarrow_S$ be an alphabetic rewriting system on $\mathpzc{S}(X,\theta)$ (not necessarily convergent). The fact that the rewriting system is alphabetic implies in an essential way the following results. 
\begin{lemma}\label{comportementdesRinvariants1}
Let $S$ be a semigroup, and let $f\colon X\rightarrow S$ be a  function that respects the commutations.  Let  $f^{\mathpzc{S}}$ be its canonical semigroup extension from $\mathpzc{S}(X,\theta)$ to $S$. If $f$ is a $R$-invariant, then for every $w,w^{\prime}\in \mathpzc{S}(X,\theta)$ such that $w\Rightarrow_R w^{\prime}$, we have $f^{\mathpzc{S}}(w)=f^{\mathpzc{S}}(w^{\prime})$. 
\end{lemma}
\begin{proof}
Since we will deal with the empty word, one needs to recall the following. 
According to lemma~\ref{adjunctionunitfpcs}, $\mathpzc{M}(X,\theta)=\mathpzc{S}(X,\theta)\cup\{\, \epsilon\,\}$, where $\epsilon$ is the empty word.  Let us define $f^{\mathpzc{S}}_1\colon \mathpzc{M}(X,\theta)\rightarrow S^1$ the canonical extension of $f^{\mathpzc{S}}$ as a monoid homomorphism. That is, whenever $w\in \mathpzc{S}(X,\theta)$, $f^{\mathpzc{S}}_1(w)=f^{\mathpzc{S}}(w)$, and $f^{\mathpzc{S}}_1(\epsilon)=1$. 
Let $w,w^{\prime}\in \mathpzc{S}(X,\theta)$ such that $w\Rightarrow_R w^{\prime}$. Then there exist $x\in X$, $w^{\prime\prime}\in \mathpzc{S}(X,\theta)$, $u,v\in \mathpzc{M}(X,\theta)$ such that $(x,w^{\prime\prime})\in R$, $w=uxv$ and $w^{\prime}=uw^{\prime\prime}v$.  Because $f$ is $R$-invariant, $f(x)=f^{\mathpzc{S}}(w^{\prime})$, and then we have $f^{\mathpzc{S}}(w)=f^{\mathpzc{S}}(uxv)=f^{\mathpzc{S}}_1(uxv)=f^{\mathpzc{S}}_1(u)f^{\mathpzc{S}}_1(x)f^{\mathpzc{S}}_1(v)=f^{\mathpzc{S}}_1(u)f^{\mathpzc{S}}(x)f^{\mathpzc{S}}_1(v)=f^{\mathpzc{S}}_1(u)f(x)f^{\mathpzc{S}}_1(v)=f^{\mathpzc{S}}_1(u)f{\mathpzc{S}}(w^{\prime\prime})f^{\mathpzc{S}}_1(v)=f^{\mathpzc{S}}_1(u)f{\mathpzc{S}}_1(w^{\prime\prime})f^{\mathpzc{S}}_1(v)=f^{\mathpzc{S}}_1(uw^{\prime\prime}v)=f^{\mathpzc{S}}_1(w^{\prime})=f^{\mathpzc{S}}(w^{\prime})$. 
\end{proof}
\begin{corollary}\label{comportementdesRinvariants2}
Let $S$ be a semigroup, and let $f\colon X\rightarrow S$ be a  function that respects the commutations. If $f$ is a $R$-invariant, then its canonical semigroup extension $f^{\mathpzc{S}}$ passes to the quotient $\Quo{\mathpzc{S}(X,\theta)}{\xLeftrightarrow{*}_R}$.  
\end{corollary}
\begin{proof}
Let $w,w^{\prime}\in \mathpzc{S}(X,\theta)$ such that $w\xLeftrightarrow{*}_R w^{\prime}$. Then there are $n>0$, $w_0,\cdots, w_n\in\mathpzc{S}(X,\theta)$, $w_0=w$, $w_n=w^{\prime}$ such that for every $0\leq i <n$, $w_i=w_{i+1}$ or $w_i\Leftrightarrow_R w_{i+1}$. Therefore for every $0\leq i<n$, either $w_i=w_{i+1}$, or $w_i \Rightarrow_R w_{i+1}$, or $w_i \Leftarrow_R w_{i+1}$. Because $f$ is a $R$-invariant, according to lemma~\ref{comportementdesRinvariants1}, for every $0\leq i<n$, $f^{\mathpzc{S}}(w_i)=f^{\mathpzc{S}}(w_{i+1})$. Therefore $f^{\mathpzc{S}}(w)=f^{\mathpzc{S}}(w_0)=\cdots=f^{\mathpzc{S}}(w_n)=f^{\mathpzc{S}}(w^{\prime})$. Then, there exists a unique semigroup homomorphism $f^{\mathpzc{S}}_{R}\colon \Quo{\mathpzc{S}(X,\theta)}{\xLeftrightarrow{*}_R}\rightarrow S$ such that $f^{\mathpzc{S}}_{R}([w]_{\xLeftrightarrow{*}_R})=f^{\mathpzc{S}}(w)$ for every $w\in \mathpzc{S}(X,\theta)$ (where $[w]_{\xLeftrightarrow{*}_R}$ denotes the class of $w$ modulo $\xLeftrightarrow{*}_R$).
\end{proof}

We are now in position to establish the main result of this paper. 
\begin{theorem}\label{mainstatement}
Let $(X,\theta)$ be a commutation alphabet, and let $(\mathpzc{S}(X,\theta),\Rightarrow_R)$ be an alphabetic rewriting system. There exist a group $\mathpzc{TG}(X,\theta,R)$ and a mapping $t\colon X\rightarrow \mathpzc{TG}(X,\theta,R)$ that respects the commutations which is an $R$-invariant such that for every group $G$, and every (commutations respecting) $R$-invariant mapping $f\colon X\rightarrow G$, there is a unique group homomorphism $h\colon \mathpzc{TG}(X,\theta,R)\rightarrow G$ such that the following diagram commutes.
 \begin{equation}
\xymatrixcolsep{1pc}
\xymatrix{
X\ar[r]^{f} \ar[d]_{t} &G\\
\mathpzc{TG}(X,\theta,R) \ar[ur]_{h} &
}
\end{equation}
Moreover, if $\Rightarrow_R$ is convergent, then the group $\mathpzc{TG}(X,\theta,R)$ is isomorphic to the free partially commutative group $\mathpzc{G}(\mathsf{Irr}(X),\theta_{\mathsf{Irr}(X)})$ and $t$ is the normal form $\mathpzc{N}\circ j_{\mathpzc{S},X}\colon X\rightarrow \mathpzc{S}(\mathpzc{Irr}(X),\theta_{\mathpzc{Irr}(X)})$ restricted to the alphabet $X$ (recall that we have $X\subseteq \mathpzc{S}(\mathpzc{Irr}(X),\theta_{\mathpzc{Irr}(X)})\subseteq \mathpzc{M}(\mathpzc{Irr}(X),\theta_{\mathpzc{Irr}(X)})\subseteq \mathpzc{G}(\mathpzc{Irr}(X),\theta_{\mathpzc{Irr}(X)})$ under natural identifications; see Subsection~\ref{fpcs}).
\end{theorem}
\begin{proof}
Let $G$ be a group  and let $f\colon X \rightarrow G$ be a commutations respecting $R$-invariant mapping. According to the universal problem of free partially commutative semigroups, because $G$ is also a semigroup, we have the following commutative diagram.
 \begin{equation}
\xymatrixcolsep{1pc}
\xymatrix{
X\ar[r]^{f} \ar[d]_{j_{\mathpzc{S},X}} &G\\
\mathpzc{S}(X,\theta) \ar[ur]_{f^{\mathpzc{S}}} &
}
\end{equation}
According to corollary~\ref{comportementdesRinvariants2}, we may complete the previous diagram in a natural way (the notations from the proof of corollary~\ref{comportementdesRinvariants2} are used).
\begin{equation}
\xymatrixcolsep{2pc}
\xymatrix{
X\ar[r]^{f} \ar[d]_{j_{\mathpzc{S},X}} &G\\
\mathpzc{S}(X,\theta) \ar[d]_{[\cdot]_{\xLeftrightarrow{*}_R}}\ar[ur]_{f^{\mathpzc{S}}} &\\
\Quo{\mathpzc{S}(X,\theta)}{\xLeftrightarrow{*}_R} \ar[ruu]_{f^{\mathpzc{S}}_R}&
}
\end{equation}
Now, we extend in a natural way $f^{\mathpzc{S}}_R$ as a monoid homomorphism $f^{\mathpzc{S}}_{R,1}\colon (\Quo{\mathpzc{S}(X,\theta)}{\xLeftrightarrow{*}_R})^1\rightarrow G$ (because $G$ is also a monoid). Let us denote by $M$ the monoid $(\Quo{\mathpzc{S}(X,\theta)}{\xLeftrightarrow{*}_R})^1$. We obtain the following diagram.
\begin{equation}
\xymatrixcolsep{8.3pc}
\xymatrix{
X\ar[r]^{f} \ar[d]_{[\cdot]_{\xLeftrightarrow{*}_R}\circ j_{\mathpzc{S},X}} &G\\
\Quo{\mathpzc{S}(X,\theta)}{\xLeftrightarrow{*}_R} \ar[d]_{i_{\mathpzc{S},\Quo{\mathpzc{S}(X,\theta)}{\xLeftrightarrow{*}_R}}}\ar[ru]_{f^{\mathpzc{S}}_R}&\\
M \ar[ruu]_{f^{\mathpzc{S}}_{R,1}}&
}
\end{equation}
Finally, using the Grothendieck group $\mathscr{G}(M)$ of $M$, we complete the previous commutative diagram as follows (where $i=i_{\mathpzc{S},\Quo{\mathpzc{S}(X,\theta)}{\xLeftrightarrow{*}_R}}\circ [\cdot]_{\xLeftrightarrow{*}_R}\circ j_{\mathpzc{S},X}$).
\begin{equation}\label{bigdiagram}
\xymatrixcolsep{5pc}
\xymatrix{
X\ar[r]^{f} \ar@/_3.5pc/[dd]_{t} \ar[d]_{i} &G\\
M \ar[d]_{i_{\mathpzc{M},M}} \ar[ru]_{f^{\mathpzc{S}}_{R,1}}&\\
\mathpzc{TG}(X,\theta,R)=\mathscr{G}(M) \ar@/_1pc/[ruu]_{\widehat{f}^{\mathpzc{S}}_{R,1}=h}& 
}
\end{equation}
Now, as illustrated in the previous diagram, let $\mathpzc{TG}(X,\theta,R)=\mathscr{G}(M)$, $t=i_{\mathpzc{M},M}\circ i_{\mathpzc{S},\Quo{\mathpzc{S}(X,\theta)}{\xLeftrightarrow{*}_R}}\circ [\cdot]_{\xLeftrightarrow{*}_R}\circ j_{\mathpzc{S},X}$, and $h=\widehat{f}^{\mathpzc{S}}_{R,1}$. First of all, $t$ obviously respects the commutations. Let us consider the canonical extension $t^{\mathpzc{S}}\colon X\rightarrow \mathpzc{S}(X,\theta)$ of $t$. So we have the following commutative diagram.
\begin{equation}
\xymatrixcolsep{7pc}
\xymatrix{
X\ar[r]^{t} \ar[d]_{j_{\mathpzc{S},X}} &\mathpzc{TG}(X,\theta,R)\\
\mathpzc{S}(X,\theta) \ar[ur]_{t^{\mathpzc{S}}} &
}
\end{equation}
By uniqueness of the solution of a universal problem, and according to the diagram~\ref{bigdiagram}, we have $t^{\mathpzc{S}}=i_{\mathpzc{M},M}\circ i_{\mathpzc{S},\Quo{\mathpzc{S}(X,\theta)}{\xLeftrightarrow{*}_R}}\circ [\cdot]_{\xLeftrightarrow{*}_R}$. Now, let $x\in X$, $w\in \mathpzc{S}(X,\theta)$ such that $(x,w)\in R$. Then, $[x]_{\xLeftrightarrow{*}_R}=[w]_{\xLeftrightarrow{*}_R}$. Therefore, $t(x)=t^{\mathpzc{S}}(x)=t^{\mathpzc{S}}(w)$, so that $t$ is $R$-invariant.
 Then the first part of the theorem is proved. 

Now, let us assume that $(\mathpzc{S}(X,\theta),\Rightarrow_R)$ is convergent. Then, by proposition~\ref{Freepartcomstruonirr}, $\Quo{\mathpzc{S}(X,\theta)}{\xLeftrightarrow{*}_R}$ is isomorphic to the free partially commutative semigroup $\mathpzc{S}(\mathsf{Irr}(X),\theta_{\mathsf{Irr}(X)})$. Therefore, $M=(\Quo{\mathpzc{S}(X,\theta)}{\xLeftrightarrow{*}_R})^1$ is isomorphic to the free partially commutative monoid $\mathpzc{M}(X,\theta)$ (by lemma~\ref{adjunctionunitfpcs}). Finally, the Grothendieck group $\mathscr{G}(M)$ is isomorphic to the Grothendieck group $\mathscr{G}(\mathpzc{M}(X,\theta))$ (because $\mathscr{G}(\cdot)$ is functorial) so that it is isomorphic to the free partially commutative group $\mathpzc{G}(X,\theta)$ (by lemma~\ref{Grothendieckoffpcg}). The fact that in this case, $t$ is the normal form $\mathpzc{N}\circ j_{\mathpzc{S},X}\colon X\rightarrow \mathpzc{S}(\mathpzc{Irr}(X),\theta_{\mathpzc{Irr}(X)})$ restricted to the alphabet $X$ (where $\mathpzc{S}(\mathpzc{Irr}(X),\theta_{\mathpzc{Irr}(X)})$ is naturally identified with a sub-semigroup of $\mathpzc{G}(\mathpzc{Irr}(X),\theta_{\mathpzc{Irr}(X)})$) is quite obvious to check.
\end{proof}
\begin{definition}
The group $\mathpzc{TG}(X,\theta,R)$ is called the \emph{Tutte-Grothendieck group} and $t$ the \emph{universal Tutte-Grothendieck $R$-invariant} of the alphabetic rewriting system $(\mathpzc{S}(X,\theta),\Rightarrow_R)$. 
\end{definition}

\subsection{Some examples}

This section is devoted to the presentation of several examples of Tutte-Grothendieck groups and universal invariants corresponding to convergent alphabetic rewriting systems. These examples come from the theory of graphs (Tutte polynomial), from algebra (Weyl algebra, and Poincar\'e-Birkhoff-Witt theorem) and from combinatorics (prefabs). 

\subsubsection{The Tutte polynomial}

In its famous paper~\cite{Tutte}, Tutte used the following decomposition of (isomorphism classes of) finite multigraphs (graphs with multiple edges and loops). Let $G$ be a multigraph, and $e$ be a link (edge which is not a loop nor a bridge) in $G$. Let $G-e$ be the graph obtained from $G$ by erasing $e$, and let $G/e$ be the graph obtained by contraction of $e$ in $G$ ($e$ is removed, and its origin and source are identified). Then $G$ is decomposed into $(G-e)+G/e$  ($+$ being the free commutative juxtaposition). As explained in~\cite{Brylawski} in terms of an order relation, a rewriting system may be defined, and the universal invariant attached to this system is the well-known Tutte polynomials (see~\cite{Tutte}).

\subsubsection{Integral Weyl algebra}
For any set $X$, let $X^{\oplus}$ be the free commutative semigroup generated by $X$ (that is, $X^{\oplus}=\mathpzc{S}(X,\theta,)$, where $\theta=(X\times X)\setminus \Delta_X$ and $\Delta_X$ is the equality relation on $X$), written additively. Recall also that the free Abelian group generated by $X$, namely $\mathpzc{G}(X,\theta)$, is isomorphic to the group (under point-wise addition) $\mathbb{Z}^{(X)}$ of all mappings from $X$ to $\mathbb{Z}$ with a finite support (the support of a function $f\colon X\rightarrow \mathbb{Z}$ is the set of all $x\in X$ such that $f(x)\not=0$), see for instance~\cite{BouAlg}.  Let $Y=\{\, a,b\,\}$ be a two element set. Let $X=Y^{*}$, and $R=\{\, (uabv,ubav + uv)\colon u,v\in Y^{*}\,\}\subseteq X\times X^{\oplus}$. It is clear that $\mathsf{Irr}(X)=\{\ b^ia^j\colon i,j\in \mathbb{N}\,\}$. Moreover the alphabetic rewriting system $(X^{\oplus},\Rightarrow_R)$ is convergent (it is not difficult to check this property using for instance techniques from~\cite{Bergman}). Let $\theta=(X\times X)\setminus \Delta_X$. Then $\mathpzc{TG}(X,\theta,R)=\mathpzc{G}(\mathsf{Irr}(X),\theta_{\mathsf{Irr}(X)})=\mathbb{Z}^{(\mathsf{Irr}(X))}$. Therefore we recover the well-known fact (see~\cite{Kassel}) that the integral Weyl algebra $A_{\mathbb{Z}}=\Quo{\mathbb{Z}\langle a,b\rangle}{I_{[a,b]}}$ with two generators (where $\mathbb{Z}\langle a,b\rangle$ denotes the ring of the free monoid $X=Y^{*}=\{\,a,b\,\}^*$, and where $I_{[a,b]}$ is the two-sided ideal of $\mathbb{Z}\langle a,b\rangle$ generated by $ab-ba-1$) is free as an Abelian group with generators $\mathsf{Irr}(X)$. The universal Tutte-Grothendieck $R$-invariant $t$ of $(X^{\oplus},\Rightarrow_R)$ is the normal form of the words in $X=Y^{*}$. For instance, $t(babab)=b^2a^2 + 3 b^2a+b$. 

Let $c$ be a variable (distinct from $a,b$) and let $Y_c=Y\cup\{\,c\,\}=\{\,a,b,c\,\}$. Consider the relation $\theta=\{\, (x,c)\colon x\in Y\,\}\cup\{\, (c,x)\colon x\in Y\,\}$. Finally let $X_c=\mathpzc{S}(Y_c,\theta)$. Therefore $c$ commutes with all elements of $X_c$. Let $R_c=\{\, (uabv,ubav + uv)\colon u,v\in \mathpzc{S}(Y_c,\theta)\,\}\subseteq X_c\times X_c^{\oplus}$. Then we can check that $(X_c^{\oplus},\Rightarrow_{R_c})$ is a convergent alphabetic rewriting system whose Tutte-Grothendieck group is $\mathbb{Z}^{\mathsf{Irr}(X_c)}$ where $\mathsf{Irr}(X_c)=\{\, c^i b^j a^k\colon i,j,k\in\mathbb{N}\,\}$ (note that $c^i b^j a^k=c^{i_1}b^j c^{i_2} a^k c^{i_3}$ for every non-negative integers $i_1$, $i_2$, $i_3$, $i=i_1+i_2+i_3$, $j$ and $k$, since $c$ commutes with all other elements). This gives us immediately a free $\mathbb{Z}$-basis for the central extension $\Quo{\mathbb{Z}\langle a,b,c\rangle}{I_{[a,b],c}}$ (where $I_{[a,b],c}$ is the two-sided ideal of the ring $\mathbb{Z}\langle a,b,c\rangle$ of the monoid $Y^{*}_c$ generated by $ab-ba-1$ and $cx-xc$ for every $x\in \{\,a,b\,\}$) of the integral Weyl algebra $A_{\mathbb{Z}}$.  

\subsubsection{The Poincar\'e-Birkhoff-Witt theorem}

Let $\mathfrak{g}$ be a Lie algebra over some basis ring\footnote{$R$ is assumed to be associative, commutative, and has a multiplicative identity.} $R$ which is free as an $R$-module (see~\cite{BouLie}). Let $\mathcal{B}$ be a basis of $\mathfrak{g}$ seen as a (free) $R$-module. Let us assume that $\mathcal{B}$ is linearly ordered by $\leq$. Let $X=\mathcal{B}^*$ be the free monoid generated by $\mathcal{B}$. Let $R=\{\, (uhgv,ughv)\colon g,h\in\mathcal{B},\ g < h,\ u,v\in \mathcal{B}^*\,\}\subseteq X\times X^{\oplus}$. It is obvious that $(X^{\oplus},\Rightarrow_R)$ is a convergent alphabetic rewriting system. Moreover, $\mathsf{Irr}(X)=\{\, g_1\cdots g_n\colon n\geq 0,\ g_i\in\mathcal{B}\ \mbox{for all}\ 0\leq i\leq n,\ g_i\leq g_{i+1}\ \mbox{for all}\ 0\leq i<n\,\}$ and its Tutte-Grotendieck group is $\mathbb{Z}^{(\mathsf{Irr}(X))}$, while its universal Tutte-Grothendieck invariant $t$ is the re-ordering of an element of $X$ in an increasing order (relative to $\leq$). We recognize the  famous Poincar\'e-Birkhoff-Witt theorem (\cite{Birkhoff,Poincare, Witt}).

\subsubsection{Prefabs}

In~\cite{BenderGoldman}, Bender and Goldman introduced the notion of a prefab, for combinatorial purposes (computation of some generating functions). We recall here (a part of) this concept. Let $X$ be a set together with a multivalued binary operation $\circ$ (meaning that $x,y\in X$ implies that $x\circ y \subseteq X$) subjected to properties given below. For every $x,y\in X$, $x\circ y$ is a finite set. The operation $\circ$ is extended to the power set $2^{X}$ of $X$ by $A\circ B=\{\ z\in S\colon x\in x\circ y\ \mbox{for some}\ x\in A,\ y\in B\,\}$. If $x\in X$ and $A\subseteq X$, then we let $x\circ A$ be equal to $\{\,x\,\}\circ A=A\circ \{\,x\,\}$, and $x^i$ is defined by induction: $x^1=\{\, x\,\}$, and $x^{i+1}=x\circ x^i$ for every positive integer $i$. We say that $(X,\circ)$ is a \emph{prefab} if the composition $\circ$ on $2^X$ is associative, commutative (therefore $2^X$ becomes a semigroup), and has an identity\footnote{The identity plays also a r\^ole in counting arguments in~\cite{BenderGoldman}.} $i\in S$ such that $x\circ i=\{\, x \,\}=i\circ x$ for every $x\in X$ (then $2^X$ is a monoid). An element $p\in X\setminus\{\,i\,\}$ is called a \emph{prime} if $p\in x\circ y$ implies $x=i$ or $y=i$. We say that $(X,\circ)$ is a \emph{unique factorization} prefab if every $x\in X\setminus\{\,i\,\}$ factors uniquely into primes in the sense that $x\in p_1^{i_1}\circ \cdots \circ p_n^{i_n}$ for a unique set of $n>0$ primes $\{\,p_i\colon 1\leq i\leq n\,\}$ and unique positive integers $i_1,\cdots,i_n$. We say that $(X,\circ)$ is a \emph{very unique factorization} prefab if $x\in \displaystyle (p_1^{i_1}\circ \cdots\circ p_m^{i_m})\circ  (q_{1}^{j_{1}}\circ \cdots\circ  q_{n}^{j_{n}})$ where $m>0,n>0$, all the $i$'s and all the $j$'s are positive integers,  all the $p$'s are mutually distinct primes, and all the $q$'s are mutually distinct primes (but some $q$'s may be equal to some $p$'s), then there exist unique elements $y\in p_1^{i_1}\circ \cdots\circ p_m^{i_m}$ and $z\in q_{1}^{j_{1}}\circ \cdots\circ  q_{n}^{j_{n}}$ such that $x\in y\circ z$. In the original definition of a prefab, there is also a mapping $f\colon 2^X\rightarrow \mathbb{N}$ which serves as a weigth function for a combinatorial use but which is not needed here. 

Let $(Y,\circ)$ be a unique and very unique factorization prefab. Let $P$ be the set of primes of this prefab. Let $X=Y\setminus\{\, i\,\}$. Let $R=\{\, (x,y+z)\colon \exists y,z\in X,\ x\in y\circ z\,\}\subseteq X\times X^{\oplus}$. According to the properties of unique factorization, very unique factorization,  associativity and commutativity of $\circ$, it is clear that $(X^{\oplus},\Rightarrow_R)$ is a convergent alphabetic rewriting system. We have $\mathsf{Irr}(X)=P$, and the Tutte-Grothendieck group is, as expected, $\mathbb{Z}^{(P)}$. It is also immediate that $t(x)=\displaystyle \sum_{j=1}^n i_j p_j$ where  $p_1^{i_1}\circ \cdots \circ p_n^{i_n}$ is the unique prime factorization of $x$.  

As examples of (unique and very unique factorization) prefabs, one can cite the two following from~\cite{BenderGoldman}. 
Let $X$ be any set, and let $w,w^{\prime}\in X^{+}$ be two words. A \emph{shuffle} of these two words  is a word $w^{\prime\prime}=x_1\cdots x_n \in X^{+}$, $x_i\in X$ for $1\leq i\leq n$ (where $n$ is the sum of lengths of $w$ and $w^{\prime}$) such that there exists $\{\, I, J\,\}$ a partition of $\{\, 1\cdots, n\,\}$  with $w=x_{i_1}\cdots x_{i_k}$, $i_1<\cdots <i_k$, $k$ is the cardinal of $I$, $I=\{\, i_1,\cdots,i_k\,\}$, and $w^{\prime}=x_{j_1}\cdots x_{j_{\ell}}$, $j_1<\cdots<j_{\ell}$, $\ell$ is the cardinal of $J$, $J=\{\, j_1,\cdots,j_{\ell}\,\}$ (such constructions appear in the shuffle product of two words;  see~\cite{Reut}). Let $w\circ w^{\prime}$ be the set of all shuffles of $w$ and $w^{\prime}$. As an example, $w=\alpha\gamma$ and $w^{\prime}=\beta\beta$. Then $w\circ w^{\prime}=\{\, \alpha\gamma\beta\beta,\alpha\beta\gamma\beta,\alpha\beta\beta\gamma,\beta\alpha\gamma\beta,\beta\alpha\beta\gamma,\beta\beta\alpha\gamma\,\}$. It is clear that the identity is the empty word (therefore we allow to choose word in $X^{*}$) while the prime elements are the letters in $X$. The prime decomposition of a word is then the set of the letters that form the words. The rewriting system associated to this prefab is the following: $(w,w^{\prime}+w^{\prime\prime})$ where $w^{\prime}+w^{\prime\prime}\in (X^{+})^{\oplus}$ such that $w\in w^{\prime}\circ w^{\prime\prime}$.  To summarize, the set $\mathsf{Irr}(X^{+})$ is $X$, the Tutte-Grothendieck group is $\mathbb{Z}^{(X)}$, and the universal invariant is given by $t(w)=\displaystyle\sum_{x\in X}|w|_x x$ (which is sometimes called the \emph{commutative image}; see~\cite{Schutz}). 

Let $x_i$ be an indeterminate for each $i\in\mathbb{N}\setminus\{\, 0\,\}$ such that $x_i\not=x_j$ whenever $i\not=j$. Let $Y=\{\ x_i\colon i \geq 1\,\}$. Let $D(x_n)=\{\, \displaystyle \sum_{i>1}k_i x_i\in Y^{\oplus}\colon  k_i\in \mathbb{N},\  \forall i>1,\ k_i=0\ \mbox{except a finite number},\ \displaystyle\sum_{i >1}k_i i = n\,\}\subseteq Y^{\oplus}$.  Finally let us define $x_m\circ x_n = \{\ f+g\in Y^{\oplus}\colon f\in D(x_m),\ g\in D(x_n) \,\}$. For instance, $x_8\circ x_4 = \{\, 6x_2, 3x_2 + x_4, x_2+2x_4, 2x_2 + x_8, x_4+x_8\,\}$. The identity is $x_1$, while the primes are exactly the $x_p$ for $p\in\mathbb{P}$, where  $\mathbb{P}$ is the set of all prime integers. Attached with these datas, the rewriting system on $(Y\setminus\{\ x_1\,\})^{\oplus}$ is given by $R=\{\, (x_n,f)\colon f\in D(x_n)\,\}$. The Tutte-Grothendieck group is $\mathbb{Z}^{(\mathbb{P})}$, and the universal invariant is given by $t(x_m)=\displaystyle \sum_{i=1}^{\ell}k_i x_{p_i}$, where $p_1^{k_1}\cdots p_{\ell}^{k_{\ell}}$ is the decomposition of $m$ into prime numbers.

\appendix
\section{A short review of Brylawski's theory}\label{shortreviewBrylawski}

In this appendix are briefly presented the main definitions and results of Brylawski's theory that are extended and clarified in this contribution.

Let $X$ be a set, and let $D(X)\subseteq X^{\oplus}$. Let $(D(X),\leq)$ be a partially ordered set with $D(X)\subseteq X^{\oplus}$ such that 
\begin{enumerate}
 \item for every $f,g\in D(X)$, if $f\leq g$, then $|f| < |g|$ or $f=g$ (where $|f|=\displaystyle\sum_{x\in X}f(x)$),
 \item for every $f,g,h\in D(X)$, $f+g\leq h$ if, and only if, there exist $h_1,h_2$ in $D(S)$ such that $h_1+h_2=h$, $f\leq h_1$, and $g\leq h_2$.   
  \end{enumerate}
  A partial ordered set of this kind is called a \emph{decomposition} of $S$, and we say that $f$ \emph{decomposes} into $g$ when $f\leq g$. An element $x$ of $X\cap D(X)$ is said to be \emph{irreducible} if $x$ is maximal with respect to $\leq$. According to axiom 1, the elements of $X$ that belong to $D(X)$ are minimal with respect to $\leq$, therefore the irreducible elements are the incomparable elements. Let us denote by $\mathsf{Irr}(X)$ their totality. A decomposition $D(X)$ is said to be \emph{finite} when for every $x\in X$, there exists $f\in D(X)\cap \mathsf{Irr}(X)^{\oplus}\subseteq X^{\oplus}$ such that $x\leq f$ (we say that $x$ \emph{fully decomposed} into $f$); in particular $X\subseteq D(X)$. A decomposition $D(X)$ is said to be \emph{refinable} if $f\leq g$ and $f\leq h$ imply that there is $k\in D(X)$ such that $g\leq k$ and $h\leq k$. By the second axiom, an element $f\in D(X)$ cannot be decomposed into any other element of $D(X)$ if, and only if, it is an element of the free commutative semigroup $\mathsf{Irr}(X)^{\oplus}$ generated by the irreducible elements, that is, a finite linear combination of irreducible elements (with non negative integer coefficients). Hence, when $D(X)$ is refinable, for each $x\in X\cap D(X)$, there is at most one way to decompose $x$ into irreducibles (that is, to fully decompose $x$). In terms of rewriting systems, it is known as the property of confluence. Finally, if $D(X)$ is both refinable and finite, then any $x\in X$ as a unique decomposition into irreducibles.  This is precisely the property of convergence of a  (Noetherian and confluent) rewriting system. Let $G$ be any Abelian group, and $D(X)$ be any decomposition of $X$. A mapping $f\colon X\rightarrow G$ is said to be an \emph{invariant} when for every $x\leq \displaystyle\sum_{i=1}^k n_i x_i$ in $D(X)$, we have $f(x)=\displaystyle\sum_{i=1}^{k}n_i f(x_i)$. We are now in position to state Brylawski's main result (to compare to theorem~\ref{mainstatement}).  
 \begin{theorem}\cite{Brylawski}
 Let $D(X)$ be a finite and refinable decomposition of $X$. There exist an Abelian group $A$ and an invariant mapping $t\colon X\rightarrow A$ such that for every Abelian group and every invariant mapping $f\colon X\rightarrow G$, there exists a unique group homomorphism $h\colon A \rightarrow G$ such that the following diagram commutes. 
 \begin{equation}
\xymatrixcolsep{1pc}
\xymatrix{
X\ar[r]^{f} \ar[d]_{t} &G\\
A \ar[ur]_{h} &
}
\end{equation}
Moreover, $A$ is freely generated by $\mathsf{Irr}(X)$. 
 \end{theorem}

\end{document}